\documentclass{article} 
\usepackage{amsthm}
\usepackage{amsmath}

\usepackage{mathtools}

\usepackage{amssymb}
\usepackage{bbm}
\usepackage{xcolor}
\usepackage{comment}
\usepackage{tikz, tikz-3dplot}
\usepackage[T1]{fontenc}

\usetikzlibrary{arrows}
\pgfarrowsdeclarecombine*{spaced stealth'}{spaced stealth'}{stealth'}{stealth'}{space}{space}
\tikzset{>=spaced stealth'}

\newcommand*{\T}{^{\mkern-1.5mu\mathsf{T}}}

\DeclareMathOperator{\cone}{cone}

\usepackage{a4}
\usepackage{todonotes} 
\usepackage{utf8math}
\usepackage{enumerate} 
\usepackage{mathrsfs}

\newcommand{\I}{\mathcal{I}}
 
\newcommand{\wt}[1]{\widetilde{ #1}}

\DeclareMathOperator{\poly}{poly}
\DeclareMathOperator{\conv}{conv}
\newcommand{\Target}{{\mbox{\boldmath$\beta$}}}
\newcommand{\boldalpha}{{\mbox{\boldmath$\alpha$}}}

\usepackage[capitalize]{cleveref}

\newcommand{\defproblem}[3]{
	\vspace{2mm}
	\vspace{1mm}
	\noindent\fbox{
		\begin{minipage}{0.95\textwidth}
			#1 \\
			{\bf{Input:}} #2  \\
			{\bf{Task:}} #3
		\end{minipage}
	}
	\vspace{2mm}
}
\newtheorem{theorem}{Theorem}
\newtheorem{lemma}[theorem]{Lemma}
\newtheorem{proposition}[theorem]{Proposition}
\newtheorem{corollary}[theorem]{Corollary}

\DeclareMathOperator{\rank}{rank}

\title{Sensitivity, Proximity and FPT Algorithms for Exact Matroid Problems}

\author{Friedrich Eisenbrand\footnote{EPFL, Lausanne, Switzerland} \and Lars Rohwedder\footnote{Maastricht University, Maastricht, Netherlands. Supported by Dutch Research Council (NWO) project “The Twilight
Zone of Efficiency: Optimality of Quasi-Polynomial Time Algorithms” [grant number OCEN.W.21.268]} \and Karol
W\k{e}grzycki\footnote{Saarland University and Max Planck Institute for Informatics,
        Saarbr\"ucken, Germany. 
    This work is part of the project TIPEA that has
    received funding from the European Research Council (ERC) under the European Unions Horizon
2020 research and innovation programme (grant agreement No. 850979).}
}

\date{}

\begin{document}
\maketitle
\begin{abstract}
  \noindent
    We consider the problem of finding a basis of a matroid with weight exactly equal to a given target.
    Here weights can be discrete values from $\{-\Delta,\ldots,\Delta\}$ or more generally $m$-dimensional
    vectors of such discrete values.
    We resolve the parameterized complexity completely, by presenting an
    FPT algorithm parameterized by $\Delta$ and $m$ for arbitrary matroids. Prior to our work, no such
    algorithms were known even when weights are in $\{0,1\}$, or arbitrary
    $\Delta$ and $m=1$.  Our main technical contributions are 
    new proximity and sensitivity bounds for matroid problems, independent of the number of elements. These bounds 
     imply FPT algorithms via matroid intersection.  
\end{abstract}

\section{Introduction}

Matroids are one of the most fundamental abstractions of combinatorial structures and capture
intricate set systems such as spanning trees and transversals while still offering tractability
for many related problems\footnote{We refer the reader to \cref{sec:preliminaries-2} and to~\cite{schrijver2003combinatorial} for basic definitions and results regarding matroids.}.
The well known Greedy algorithm can find a minimum (or maximum) weight basis of a matroid.
Inherently more difficult is the task of finding a basis $B$ with weight $w(B) = \sum_{b\in B} w(b)$
exactly equal to some given target $\beta \in \mathbb Z$.
A straightforward  reduction from Subset Sum shows that the problem is weakly NP-hard even for the most trivial examples of
matroids.
Papadimitriou and Yannakakis~\cite{papadimitriou1982complexity} first mention this and observe
that for $0,1$ weights the problem can still be solved efficiently via matroid intersection.
They also mention that this generalizes to a fixed number of equality constraints,
that is, given $m$-dimensional integral weight vectors $W(e) \in \{-\Delta,\dotsc,\Delta\}^m$ for each element $e$
and a target $\Target\in\mathbb Z^m$ the goal is to find a basis $B$ with $W(B) = \sum_{b\in B}W(b) = \Target$. 
Towards this, one can guess the correct number of elements for each distinct weight vector in time $O(n^{(2\Delta+1)^m})$ and then
solve matroid intersection in polynomial time where an additional partition matroid dictates the correct number of
elements of each weight vector.
Papadimitriou and Yannakakis also asked whether a pseudopolynomial time algorithm exists for
spanning trees. Via an algebraic algorithm that uses a variant of Kirchhoff's theorem
this is indeed possible~\cite{barahona1987exact} and generalizes to all linear
matroids~\cite{camerini1992random}.
This leads to an algorithm with running time $(n\Delta)^{O(m)}$ in the setting mentioned above.
Algebraic methods to solve exact weight problems were also mentioned by Mulmuley, Vazirani, and Vazirani~\cite{mulmuley1987matching}, who credit Lov\'asz for describing an algorithm for the exact matching problem.
In contrast, Doron-Arad, Kulik, and Shachnai~\cite{doronarad2023tight} very recently showed that the problem
cannot be solved in
pseudopolynomial time for arbitrary matroids (in the standard independence oracle model).

Similar settings as the above have also been studied extensively in the field of  approximation algorithms. 
Grandoni and Zenklusen~\cite{grandoni2010approximation} consider a range of multi-budgeted matroid problems, i.e.,
imposing linear inequalities with non-negative weights on different variants.
As they observe, the problem of finding any basis subject to
two such constraints is already weakly NP-hard. Thus, their focus is on finding any independent set and not
only bases.
For the problem of finding a maximum profit independent set subject to a fixed number of budget constraints
Grandoni and Zenklusen derive a PTAS.
EPTAS or FPTAS algorithms cannot exist due to a hardness result for 2-dimensional Knapsack~\cite{kulik2010there}.
Further, an EPTAS is known for a single budget constraint~\cite{doron2023eptasb,doron2023eptas}.
We note that while in some problems, most famously Knapsack, approximation schemes can easily be derived from
pseudopolynomial time algorithms, this is not true for multibudgeted independent set.

For the type of problems we mention above, the classical complexity and approximation algorithms are very well understood by now, but
the status through the lens of parameterized algorithms\footnote{We refer to~\cite{cygan2015parameterized} for background on parameterized (FPT) algorithms.}
is still unsatisfying with answers being unknown even for the following basic questions:
\begin{quote}
	For $0, 1$ weights and graphic matroids, is there an FPT algorithm in parameter $m$?

	For a single equality constraint, is there an FPT algorithm in parameter $\Delta$ for arbitrary matroids?
\end{quote}
We resolve the parameterized complexity completely by providing an FPT algorithm in parameters $\Delta$ and $m$ for
arbitrary matroids. This is an algorithm with a running time of the form $f(Δ,m) ⋅ \poly(n)$, where $f(Δ,m)$ is a function depending on these parameters $Δ$ and $m$ only,~see, e.g.~\cite{cygan2015parameterized}. In our case, $f(Δ,m) = (mΔ)^{O(Δ)^m}$.
\medskip 

\noindent 
We remark the connection to binary integer linear programming of the form
\begin{equation}
	\{x\in\{0, 1\}^n : Ax = b\},\label{eq:ilp}
\end{equation}
where $A\in \{-\Delta,\dotsc,\Delta\}^{m\times n}$ and $b\in\mathbb Z^m$.
This is a very simple example of a multidimensional exact matroid problem over the uniform matroid
with $2n$ elements and rank~$n$: elements $1,2,\dotsc,n$ correspond to the variables $x_1,\dotsc,x_n$ with
weight vector $W(i) = A_i$ and elements $n+1,\dotsc,2n$ that have a zero weight vector and are used to ensure
a solution is a basis. Despite the simplicity of the matroid, it is already non-trivial to find FPT algorithms
for integer~program~\eqref{eq:ilp} in parameters $\Delta$ and $m$.
Such results were obtained by Papadimitriou~\cite{papadimitriou1981complexity},
using a slightly stronger parameterization, and by Eisenbrand and Weismantel~\cite{eisenbrand2019proximity}
with only parameters $\Delta$ and $m$.
The latter work is heavily based on proximity and sensitivity, which turn out to be the key elements also in our work.
Throughout this document, we use the terms
proximity and sensitivity to describe the distance between
an optimal continuous solution and an optimal integer solution and between two integer solutions
with a similar right-hand side. Bounds on these quantities are useful algorithmically,
specifically to reduce the search space, but are also of independent interest, for example, to understand
how severe the effects of uncertainty in data can be for decision making, see e.g.~\cite{insua1991framework,triantaphyllou1997sensitivity}.
The study of general proximity and sensitivity bounds in integer linear programming goes back to
the seminal work by Cook, Gerards, Schrijver, and Tardos~\cite{cook1986sensitivity}.

\subsection{Our contribution}
\begin{figure}
	\centering
     \tdplotsetmaincoords{70}{100}
	\begin{tikzpicture}[scale=0.9, tdplot_main_coords, line join=round]
		\draw[thick] (-3.5, 3.5, 0) -- (-3.5, -3.5, 0);

		   \pgfmathsetmacro\a{1}
    \pgfmathsetmacro{\phi}{\a*(1+sqrt(5))/2}
    \path
    coordinate(A) at (0,\phi,\a)
    coordinate(B) at (0,\phi,-\a)
    coordinate(C) at (0,-\phi,\a)
    coordinate(D) at (0,-\phi,-\a)
    coordinate(E) at (\a,0,\phi)
    coordinate(F) at (\a,0,-\phi)
    coordinate(G) at (-\a,0,\phi)
    coordinate(H) at (-\a,0,-\phi)
    coordinate(I) at (\phi,\a,0)
    coordinate(J) at (\phi,-\a,0)
    coordinate(K) at (-\phi,\a,0)
    coordinate(L) at (-\phi,-\a,0)
    coordinate(N) at (0,\phi,0)
    coordinate(M) at (0,-\phi,0);

        \fill[lightgray]
        (A) -- (I) -- (B) --cycle
        (F) -- (I) -- (B) --cycle
        (F) -- (I) -- (J) --cycle
        (F) -- (D) -- (J) --cycle
        (C) -- (D) -- (J) --cycle
        (C) -- (E) -- (J) --cycle
        (I) -- (E) -- (J) --cycle
        (I) -- (E) -- (A) --cycle
        (G) -- (E) -- (A) --cycle
        (G) -- (E) -- (C) --cycle;
		
	\fill[gray] (N) -- (I) -- (J) -- (M) -- (L) -- (K) -- cycle;
	\draw[] (N) -- (I) -- (J) -- (M);

        \draw[]
        (A) -- (I) -- (B) --cycle
        (F) -- (I) -- (B) --cycle
        (F) -- (I) -- (J) --cycle
        (F) -- (D) -- (J) --cycle
        (C) -- (D) -- (J) --cycle
        (C) -- (E) -- (J) --cycle
        (I) -- (E) -- (J) --cycle
        (I) -- (E) -- (A) --cycle
        (G) -- (E) -- (A) --cycle
        (G) -- (E) -- (C) --cycle;

		\draw[thick] (-3.5, 3.5, 0) -- (3.5, 3.5, 0) -- (3.5, -3.5, 0) -- (-3.5, -3.5, 0);

		\node at (0,\phi + 0.25,0) {$x^*$};
		\node at (0,\phi + 0.25, - \a) {$B$};
		\node at (\phi,\a, 0.5) {$A'$};
		\node at (\phi,-\a, 0.5) {$A$};
		\node at (0, 4.5, 0) {$W x = \Target$};
		\node at (0, 1.5, 2) {$P_B(M)$};
	\end{tikzpicture}
	\caption{Schematic overview of proximity, sensitivity and their connection. The vertices of $P_B(M)$ (light gray) are the bases of $M$. The intersection (dark gray) of $P_B(M)$ with the affine subspace $\{x \in \mathbb R^E : Wx = \Target\}$ may contain non-integral vertices. For such a vertex $x^*$ by standard rounding there always exists a close by basis $B$ with $W(B) \approx \Target$. If there is a basis $A$ with $W(A) = \Target$ and the distance to $B$ (equivalently, to $x^*$) is sufficiently large, then by sensitivity there is a closer basis $A'$ also with $W(A') = \Target$. This implies proximity.}
	\label{fig:proximity}
\end{figure}
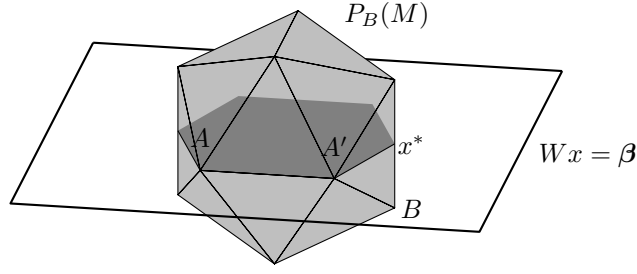
Let $M = (E, \I)$ be a matroid with (possibly multidimensional) weights $W(e)\in\{-\Delta,\dotsc,\Delta\}^m$.
In the following, we denote by $A\oplus B = (A\setminus B) \cup (B\setminus A)$ the symmetric difference of sets $A$ and $B$.
We prove the following sensitivity result.
\begin{theorem}[Sensitivity Theorem for Matroids] \label{thr:5}
	Let $A, B$ be bases of $M$. Then there exists a basis $A'$ with $W(A') = W(A)$ and
	\begin{equation*}
		|A'\oplus B| \le (2 m Δ)^{12 m}  \cdot \|W(B) - W(A)\|_1 .
	\end{equation*}
      \end{theorem}
\noindent       
Denote by $P_B(M)\subseteq [0, 1]^E$
the matroid base polytope, that is, the convex hull of indicator vectors $\chi(B)$ of the bases $B$ of $M$.
In particular, there is a one-to-one correspondence between integral elements of $P_B(M)$ and bases of $M$.
For convenience, we write $W ∈ ℤ^{m × n}$ as the matrix with columns $W(e)$ in the order the elements $e$ appear
as dimensions in $P_B(M)$. For $S ⊆ E$ we write $W(S) = ∑_{e ∈S} W(e)$. 
We prove the following proximity result, see also \cref{fig:proximity}.
\begin{theorem}[Proximity Theorem for Matroids]\label{thm:proximity}
	Let $A$ be a basis of $M$ and
	let $x^*$ be any vertex solution to the polytope 
	\begin{equation*}
		\{ x ∈ ℝ^n ： x\in P_B(M), \, Wx = W(A)\} .
              \end{equation*}
	There exists a basis $A'$ of $M$ that satisfies $W(A') = W(A)$ and               
	\begin{equation*}
		\| x^* - χ(A') \|_1 \le (2mΔ)^{13 m} .
	\end{equation*}
\end{theorem}
These sensitivity and proximity bounds are in the same order of magnitude as those known for~\eqref{eq:ilp},
see e.g.~\cite{eisenbrand2019proximity}.
From the proximity theorem we derive the following FPT algorithms.
\begin{theorem}\label{thm:alg}
	For target $\Target\in\mathbb Z^m$
	there is an algorithm that finds a basis $A$  of $M$ with $W(A) = \Target$, if one exists,
    in time
	\begin{equation*}
		\Delta^{O(\Delta)^m} \cdot n^{O(1)} \ .
	\end{equation*}
	Furthermore, if $M$ is linear (with a given representation), it can be
	improved to 
	\begin{equation*}
		(m\Delta)^{O(m^2)} \cdot n^{O(1)}
	\end{equation*}
    randomized time.
\end{theorem}
Due to its generality, \cref{thm:alg} can be used to obtain FPT algorithms for many concrete
applications. To name a few, it improves an FPT algorithm for Feedback Edge Set with Budget
Vectors due to Marx \cite{marx2009parameterized},
generalizes a recent algorithm proposed by Liu and Xu
\cite{liu2023congruency} for Group-Constrained Matroid Base to arbitrary finite groups,
and generalizes so-called combinatorial $n$-fold integer programs~\cite{knop2020combinatorial}, which have a wide range of applications
themselves.
We refer to \cref{sec:applications} for details.

There is a simple example based on long even cycles that shows that both sensitivity and proximity
are unbounded for bipartite perfect matching, which is a special case of matroid intersection,
even with a single $0, 1$ weight constraint. For details, see \cref{sec:intersection}.

We remark that the proximity bounds in~\cite{cook1986sensitivity,eisenbrand2019proximity} also
hold in the optimization version, i.e., when we measure for some linear optimization direction the
distance between optimal continuous solution and the closest optimal integer solution,
and this optimization direction can have arbitrary real coefficients, i.e., they are not necessarily discrete
as the equality constraints. This is interesting also from a polyhedral perspective.
It remains open whether the optimization variant also admits the strong proximity in our setting.
While this is not obviously connected to this question, there are also other disparities between feasibility
and optimization in exact matroid problems. Namely, the algebraic techniques, which are the only
known approach to solve exact matroid basis in pseudopolynomial time on linear matroids (and even graphical matroids),
are also not capable of optimization.

\subsection{Techniques}
Our main technical contribution lies in the sensitivity and proximity bounds.
We will briefly review previous techniques, specifically those by
Cook et al.~\cite{cook1986sensitivity} and by Eisenbrand and Weismantel~\cite{eisenbrand2019proximity}.
Suppose for some solution $x$ to an integer linear program we want to prove a bound on the distance to the
closest solution $z$, where the right-hand side is slightly perturbed.
One can naturally decompose the change from $z$ to $x$ into atomic changes that involve the increase or
decrease of a variable by $1$.
On a high level, Cook et al.'s approach can be summarized as defining small bundles of atomic changes that
do not affect feasibility.
In their case they use Hilbert basis elements of a carefully chosen cone and write $x - z$ as
a (not necessarily integral) conic combination
of these bundles. If the distance
between $z$ and $x$ is sufficiently large, one of the bundles is taken at least once and then we can also apply only
this bundle exactly once to $z$, proving that there is a closer solution.
Crucial to this argument is that each bundle of atomic changes can be applied
to $z$ independently.
In an inherently different approach, Eisenbrand and Weismantel arrange the atomic changes in a careful sequence
given by the Steinitz Lemma and achieve that if the distance between $x$ and $z$ is sufficiently large,
by pigeonhole principle there will be a ``cycle'' in the sequence, which is
an interval of atomic changes that does not affect feasibility.
This cycle is applied to $z$, which then also proves that there is a closer solution to $x$.

In the context of matroids, a natural candidate for an atomic change is the exchange of a pair of elements.
Both of the approaches above are applicable to restrictive classes of matroids, specifically \emph{strongly
base-orderable} matroids, see~\cite[42.6c]{schrijver2003combinatorial}.
Essentially,   this limited class 
allows any subset of the atomic changes to be applied to $z$ in isolation. 
Implementing this approach to general matroids seems to be elusive. 

Our proof is based on a novel  structural result on matroids. If two independent sets $A$ and $B$ are of the same cardinality and are roughly of the same weight, then there are large \emph{unicolor} subsets of them of equal cardinality, that can be mutually exchanged (\cref{thr:1} and \cref{co:1}). A set is unicolor if all elements have the same weight. This is a key notion of this paper.  Via identifying several such exchanges that nullify the total change of weight, and the theory of matroid intersection~\cite{edmonds1979matroid,lawler1975matroid}, we can then derive the existence of another independent set of the matroid of equal cardinality that has the same weight as $A$. Specifically in higher dimension, we furthermore rely on polyhedral combinatorics like Carathéodory's theorem and bounds on the complexity of facets and vertices. 

\section{Preliminaries}  
\label{sec:preliminaries-2}

A \emph{matroid}  $M = (E,\I)$ is defined by a \emph{ground set} $E = \{1,2,\dotsc,n\}$ and a collection $\I
\subseteq 2^{E}$ of independent sets that satisfy the following properties:
\begin{enumerate}[(M1)]
\item $\emptyset \in \I$. \label{item:3}
\item If $X \subseteq Y$ and $Y \in \I$, then $X \in \I$. \label{item:4}
\item If $X, Y \in \I$ and $|X| < |Y|$, then there exists $e \in Y\setminus X$ such that $X \cup \{e\} \in \I$. \label{item:5} 
\end{enumerate}
Condition~(M\ref{item:5}) is the \emph{exchange property}. We often use the following elementary consequence.
\begin{lemma}[Downsizing] 
  \label{lem:3}
  Let $I ∈ \I$ be an independent set $A ⊆ I$ and $B ⊆ E \setminus I$
  of equal cardinality $|B| = |A|$ and suppose that
  $(I\setminus A) ∪B ∈ \I$.  Then, for each $A'⊆A$ there exists
  $B' ⊆B$ with
  \begin{enumerate}[i)]
  \item $|A'| = |B'|$ and  \label{item:1}
  \item  $(I \setminus A') ∪B' ∈ \I$.  \label{item:2}
  \end{enumerate}
  Similarly, for each $B'⊆B$ there exists $A' ⊆A$  satisfying \ref{item:1}) and \ref{item:2})  
\end{lemma}
\begin{proof}
The set $(I\setminus A')$ is an independent set.
	We can identify $|A'|$ elements of $((I\setminus A) \cup B) \setminus (I\setminus A') = B$ that can be added to this set. 
  Similarly, $(I\setminus A) ∪B'$ is an independent set
	and can be extended to the cardinality of $I$. 
\end{proof}

\noindent 
A \emph{basis} of $M$ is an inclusionwise maximal set of $\I$. All the bases of a matroid have the same cardinality, denoted by $\rank(M)$,  the \emph{rank} of  $M$.
For a subset $S ⊆ E$, its  \emph{characteristic vector}  $χ(S) ∈ \{0,1\}^n$ is defined as
\begin{displaymath}
  χ(S)_i  =
  \begin{cases}
    1, \text{ if } i ∈ S,\\ 
    0 \text{ otherwise}. 
  \end{cases}
\end{displaymath}
The \emph{rank} of a subset  $S⊆ E$ is the rank of the matroid $(S, \I')$, where $\I' = \{ I ∩ S： I ∈ \I\}$ and we denote it by $\rank(S)$. 
The \emph{matroid polytope} $P(M) ⊆ ℝ^n$ is the convex hull of the
characteristic vectors $χ(A) ∈ \{0,1\}^n$ of independent sets $A ∈\I$.
Edmonds~\cite{edmonds1970submodular} has shown that $P(M)$  is described by the following set of inequalities  
\begin{equation}
  \label{eq:9}
    \begin{array}{ll}
      ∑_{e ∈ S} x_e  ≤   \rank(S), & S ⊆ E \\ 
                 x_e  ≥  0, &  e ∈ E.                  
    \end{array}
  \end{equation}
  The convex hull of the characteristic vectors of bases of $M$ is obtained from~\eqref{eq:9} by adding the equation $∑_{e ∈ E} x_e = \rank(M)$ and we denote this  \emph{base polytope} by $P_B(M)$. 
  
\medskip
\noindent
An \emph{($m$-dimensional) weight} of a matroid is given by a matrix $W ∈ ℤ^{m × n}$.  The weight of a subset $S ⊆ E$ is defined as $W(S) = W ⋅ χ(S)$. A \emph{target} is an integral vector $\Target ∈ ℤ^m$. A subset $S ⊆ E$ is \emph{exact} for $W$ and $\Target$, if $W(S) = \Target$. We also refer to the condition $W(S) = \Target$ as imposing $m$ \emph{constraints}. The largest absolute value of an entry of the matrix $W ∈ ℤ^{m ×n}$ is denoted by $\|W\|_∞$. 
  
\section{FPT Algorithms}
\label{sec:algorithms}

In this section we discuss how the proximity bound in Theorem~\ref{thm:proximity} can be used to derive the FPT algorithms, thereby proving Theorem~\ref{thm:alg}. 
Let $M = (E,\I)$ be a matroid, $W ∈ \{-Δ,\dots,Δ\}^{m ×n}$ be a matrix and $\Target ∈ ℤ^m$ be a target vector. 
The goal is to 
  find  a basis $B ∈ \I$ of the matroid $M$ with
  \begin{equation*}
    \label{eq:8}  
    W(B) = \Target ,
  \end{equation*}
  or to assert that such a basis does not exist.
Here $W(B)$ is the sum of all columns of the matrix $W$ that correspond to elements of $B$. Thus $W(B) = W ⋅ χ(B)$, where $χ(B) ∈ \{0,1\}^n$ denotes characteristic vector of the basis $B$.

\medskip 

We first assume that the matroid $M$ is \emph{implicitly  given} by an
\emph{independence testing oracle}~\cite[Section
40.1]{schrijver2003combinatorial}.  
 This means that testing  $S ∈ \I$  for a subset $S ⊆ E$ can be done in constant time and no other information on $M$ can be queried.  At the end of this section, we treat the second part of Theorem~\ref{thm:alg}, where a linear matroid with explicit representation given and faster algorithms can be derived. 

\medskip   
With the ellipsoid method, we  compute a vertex $x^*$ of the base polytope $P_B(M)$ intersected with the subspace $\{ x ∈ ℝ^n ： W x = \Target\}$
in time $\poly(n + m\log |Δ|)$ using~\cite{khachiyan1980polynomial,grotschel2012geometric}. Recall that the base polytope is the convex hull of incidence vectors of bases $P_B(M)  = \conv \{ χ(B)： B \text{ basis of }M\}$.
  Theorem~\ref{thm:proximity} implies that, if there exists a basis
  $A ∈ \I$ of the matroid $M$ with weight $W(A) = \Target$, then there
  exists a such basis with
  \begin{equation}
    \label{eq:7-new}    
    \|χ(A) - x^* \|_1 ≤ (m\Delta)^{O(m)}.
  \end{equation}  
  For $\boldalpha ∈ \{-Δ,\dots,Δ\}^m$ and $S ⊆ E$, let $ℓ_\boldalpha(S) ∈ ℕ_0$ denote the number of elements of weight $\boldalpha$ in $S$. In other words,
  \begin{displaymath}
    ℓ_\boldalpha(S) = | \{ e ∈ E ： W(e) = \boldalpha\} |.
  \end{displaymath}
  The proximity condition~\eqref{eq:7-new} implies
  \begin{align*}
	  ∑_{e ∈ E ： W(e) = \boldalpha} \hspace{-1em} x^*_e - (mΔ)^{O(m)} &≤ ℓ_\boldalpha(A) \\
	  &≤ ∑_{e ∈ E ： W(e) = \boldalpha} \hspace{-1em} x^*_e + (mΔ)^{O(m)} .
\end{align*}
We can guess the correct value $ℓ_\boldalpha(A)$ out of a set of
$(mΔ)^{O(m)}$ candidates for each $\boldalpha$. Since $\boldalpha ∈ \{-Δ,\dots,Δ\}^m$ this leaves
a total number of $(m\Delta)^{O(m) \cdot O(\Delta)^m} = \Delta^{O(\Delta)^m}$ vectors among one encodes
the values $ℓ_\boldalpha(A)$ for each $\boldalpha$. 
Assume that we have one such candidate. Observe that the condition
\begin{displaymath}
  ∑_{\boldalpha} ℓ_\boldalpha(A) = \rank(M)
\end{displaymath}
must be satisfied. We next consider the partition matroid $M_p = (E, \I_p)$ with independent sets
\begin{displaymath}
  \I = \left\{ S ⊆ E \, ：ℓ_\boldalpha(S) ≤ ℓ_\boldalpha,\boldalpha\in\{-\Delta,\dotsc,\Delta\}^m \right\}.
\end{displaymath}
We are looking for  a basis  $A$ of $M$ that is also contained in $M_p$. This is a \emph{matroid intersection} problem for  $M ∩M_p$ and can be solved in  time polynomial in $n$~\cite{edmonds1979matroid,lawler1975matroid}, see also~\cite{schrijver2003combinatorial}.  All together, this shows that we can find a basis $B$ of weight $W(B) = \Target$ or assert that no such basis exists, in time  $Δ^{O(Δ)^m} n^{O(1)}$.  This proves the first part of Theorem~\ref{thm:alg}.

\subsection{Linear matroids}
\emph{Linear matroids} are matroids that can be defined by a matrix $A$ over a
field $\mathbb{F}$, such that the ground set $E$ is the set of columns of $A$
and $X$ is an independent set if these columns are linearly independent. If a
matroid can be defined by a matrix $A$ over a field $\mathbb{F}$, then we say
that the matroid is \emph{representable} over $\mathbb{F}$. In this paper, when
we consider linear matroids, we always assume that the matrix $A$ that
represents the matroid is also given on the input (along the independence
testing oracle). We assume, that $\mathbb{F}$ is either a finite field or the
rationals. For a detailed discussion regarding representation issues and its
computational complexity we refer the reader to~\cite[Section
3]{marx2009parameterized}.

Camerini et
al.~\cite{camerini1992random} presented a randomized, pseudopolynomial
time algorithm to find a basis $B$ of a linear matroid with
weight $w(B)$ exactly $\beta$ in $(\Delta n)^{O(1)}$ time. In our setting, this is the one-dimensional case, i.e., the case $m=1$. This method does not immediately yield an FPT algorithm in $m$ and $\Delta$, but it can be
combined with our proximity bound leading to one and improving on the running time 
of the previous method. This uses the pseudopolynomial time algorithm as a black box and
for matroids where a deterministic variant is known, e.g. for graphic matroids~\cite{barahona1987exact},
our algorithm is also deterministic.
However, for general matroids given by an independence testing oracle, it is known that
a pseudopolynomial time algorithm cannot exist~\cite{doronarad2023tight}. 

We use a variant of a standard method of aggregating all $m$ constraints into one single constraint $w(B) = β$, see, e.g.~\cite{kannan1983polynomial},
which can be done more efficiently when the search space is bounded due to proximity.
To this end, suppose there is an algorithm that, for a given matroid $M = (E, \I)$ with weights $w : E\rightarrow\{-\Delta,\dotsc,\Delta\}$ and a target $\alpha\in\mathbb N$ finds a basis $B$ of $M$ with $w(B) = \alpha$ in time $(n\Delta)^{O(1)}$.
We will show that for the same matroid $M$, with a multidimensional weight $W : E\rightarrow \{-\Delta,\dotsc,\Delta\}^m$
and a target $\Target\in \mathbb Z^m$ one can find a basis $B$ with $W(B) = \Target$ in time
$(m\Delta)^{O(m^2)} n^{O(1)}$.
Let $x^*$ be a vertex to the matroid base polytope $P_B(M)$ intersected with $\{x \in\mathbb R^E : Wx = \Target\}$.
Let $A$ be the basis with $W(A) = \Target$ that is close to $x^*$ as guaranteed by \cref{thm:proximity}.
We write $\lfloor x^* \rceil$ for the
vector derived from $x^*$ by rounding each component to the closest
integer and set
\begin{displaymath}
	    \Gamma := \|\chi(A) - \lfloor x^*\rceil \|_1 \le 2 \|\chi(A) - x^* \|_1 \le (m\Delta)^{O(m)} .
\end{displaymath}
The first inequality follows from the fact that when
$|\chi(A)_i - x^*_i| < 1/2$ then rounding will decrease the distance
in this dimension, and otherwise, it will increase it by at most $1$.
Since we do not know $A$, we also do not know the value of $\Gamma$,
but we can obtain it through guessing.

Let $B$ be any basis with $\|\chi(B) - \lfloor x^*\rceil \|_1 = \Gamma$ and define 
\begin{displaymath}
  \lambda := (1,
  (2\Gamma \Delta +1), (2\Gamma\Delta +1)^2,\dotsc,(2\Gamma \Delta+1)^{m-1})\in
    \mathbb Z^m. 
    \end{displaymath}
We argue that
$W(B) = \Target$ if and only if $\lambda\T W(B) = \lambda\T\Target$.
Since the other direction is trivial, assume that
$\lambda\T W(B) = \lambda\T\Target$.  Inductively, one can conclude
that also $W(B)_i = \Target_i$ for $i=1,2,\dotsc,m$: because all
$j < i$ satisfy this equality by induction and all $j > i$ are
multiplied by a higher power of $2\Gamma\Delta + 1$ than constraint
$i$. Hence, one has
\begin{displaymath}
  W(B)_i \equiv \Target_i \mod 2\Gamma\Delta + 1 .
\end{displaymath}
Further, $W(B)_i$ and $\Target_i = W(A)_i$ are both in
$\{W \lfloor x^* \rceil - \Gamma \Delta, \dotsc, W \lfloor x^* \rceil
+ \Gamma \Delta\}$ and therefore the modulo operator is a bijection
and it follows that $W(B)_i = \Target_i$.

It is not enough to run the pseudopolynomial time algorithm with
$\lambda\T W(B) = \lambda\T \Target$, since we also need to enforce
that $\|\chi(B) - \lfloor x^*\rceil \|_1 = \Gamma$.  Towards this, let
$w(e) = w_1(e) + (2n + 1) w_2(e)$ be the one-dimensional weight of
element $e$, where
\begin{displaymath}
  w_1(e) = \left(1-2 \lfloor x^*_e \rceil \right) \text{ and } w_2(e) = \lambda\T W(e) .
\end{displaymath}
Further, let
\begin{displaymath}
  \alpha := \left(\Gamma - \| \lfloor x^* \rceil \|_1 \right)  + (2n+1) \cdot \lambda\T \Target 
\end{displaymath}
be the target weight. Let $B$ be a basis. We argue that
$w(B) = \alpha$ if and only if
$\|\chi(A) - \lfloor x^* \rceil \|_1 = \Gamma$ and $W(B) = \Target$.
One has the following connection between $w_1(B)$ and
$\|\chi(B) - \lfloor x^* \rceil \|_1$.
\begin{align*}
  w_1(B) + \| \lfloor x^* \rceil \|_1 &= \sum_{e\in E} \chi(B)_e \cdot (1-2 \lfloor x^*_e \rceil) + \lfloor x^*_e \rceil \\ &= \sum_{e\in E} \chi(B)_e + \lfloor x^*_e \rceil - 2\chi(B)_e \lfloor x^*_e \rceil \\
	&= \| \chi(B) - \lfloor x^*_e \rceil \|_1 ,
\end{align*}
where the last inequality follows because $a+b-2ab = |a-b|$ for
$a,b\in \{0,1\}$. Thus, $\|\chi(A) - \lfloor x^* \rceil \|_1 = \Gamma$
and $W(B) = \Target$ implies $w(B) = \alpha$. For the other direction
assume that $w(B) = \alpha$. Then
    \begin{displaymath}
	    w_1(B) \equiv w(B) \equiv \beta \equiv \Gamma - \| \lfloor x^* \rceil \|_1 \mod 2n + 1 .
    \end{displaymath}
    Since both sides are in $\{-n,\dotsc,n\}$, it follows that $w_1(B) = \Gamma - \| \lfloor x^* \rceil \|_1$
    and further $\lambda\T W(B) = w_2(B) = \lambda\T \Target$, which means that $W(B) = \Target$.
    Thus, it suffices to run the pseudopolynomial time algorithm with $w$ and $\alpha$. The maximum weight given to the 
    algorithm is $n \cdot O(\Delta\Gamma)^{m} = n \cdot (m\Delta)^{O(m^2)}$, which leads to the claimed running time.

\section{From sensitivity to proximity}
\label{sec:from-sens-prox}

In this section, we show how to obtain the proximity bound in
Theorem~\ref{thm:proximity} from Theorem~\ref{thr:5}.  We assume that
$W ∈ ℤ^{m ×n}$ is a matrix with $\|W\|_∞ ≤ Δ$.  Let $A$ be a basis of
$M = (E,\I)$ and $\Target = W(A)∈ ℤ^m$. Furthermore, let $x^*∈ ℝ^n$ be
a vertex of the polytope
\begin{equation}
  \label{eq:10}
  P_B(M) ∩ \{ W x = \Target \}.
\end{equation}
The goal is to show that there exists a basis $A'$ of $M$ with weight $W(A') = \Target$ that is close to $x^*$.

The next Lemma shows something weaker, namely that there exists a
basis close to $x^*$ whose weight might violate the target
$\Target$. But since it is close, it does not violate this target by
much. The lemma essentially follows from the fact that the
characteristic vectors of two bases $χ(B_1)$ and $χ(B_2)$ of $M$ are
neighboring vertices of $P_B(M)$ if and only if
$|B_1 ∩ B_2| = \rank(M)-1$, see,
e.g.~\cite[Theorem~40.6]{schrijver2003combinatorial}.

\begin{lemma}
  \label{lem:1}
  Let $x^* ∈ P_B(M)$ and $F ⊆ P_B(M)$ be the unique face of $P_B(M)$
  of minimal affine dimension containing $x^*$. Suppose $\dim(F) =
  d$. There exists a basis $B$ of $M$ with $χ(B) ∈ F$ with
  \begin{displaymath}
    \|χ(B) - x^* \|_1 ≤ d. 
  \end{displaymath}
\end{lemma}
\begin{proof}
  The proof is by induction on $\dim(F) = d$. If $d = 0$, then $x$ is
  an integer vector and therefore the characteristic vector of a
  basis.  Now, we suppose that $d≥1$. The face $F$ is defined by
  setting all inequalities of the base polytope to equality that are
  satisfied by $x^*$ with equality. The point $x^*$ lies in the
  relative interior of $F$. This means that it satisfies all other
  inequalities of $P_B(M)$ strictly.

  The face $F$ contains two different adjacent vertices $u,v ∈ F∩ ℤ^n$
  of $P_B(M)$. Their difference $u-v ∈ \{0,\pm 1\}^n$ satisfies
  $\|u-v\|_1 = 2$, see,
  e.g.~\cite[Theorem~40.6]{schrijver2003combinatorial}.  We can assume
  that $\|u - x^*\|_1\le \|v - x^*\|_1$ holds, otherwise swap $u$ and $v$.
  For $λ>0$ small enough, one has $x^* + λ(u-v) ∈ F$. Let $λ'>0$ be
  maximal with $x^* + λ'(u-v) ∈ F$. Since also $x^* - λ'(u-v) ∈ F$ and
  $F⊆ [0,1]^n$ one has $λ' ≤ 1/2$. Furthermore, the point
  $y' = x^* + λ'(u-v) $ lies on a face $F' ⊂ F$ of $P_B(M)$ of dimension
  $\dim(F') ≤ d-1$. By induction, there exists an integer point
  $z ∈ F' ∩ℤ^n$ with 
  $\|z - y'\|_1 ≤ d-1$. The triangle inequality implies that the basis $B$ with  $χ(B) = z$ satisfies the assertion. 
\end{proof}

\begin{proof}[Proof of Theorem~\ref{thm:proximity}]  
  Since $x^*$ is a vertex solution, it must lie on a $d$ dimensional
  face of $P_B(M)$ with $d\le m$.  By \cref{lem:1} there is an
  integer point $χ(B)\in P_B(M) \cap \mathbb Z^n$ with
  $\|χ(B) - x^*\|_1 \le d \le m$. 
  Thus, by
  \cref{thr:5} there exists a basis $A'$ with $W(A') = W(A)$ and $|A' ⊕ B| ≤ (2mΔ)^{12 m}  \|W(B) - W(A)\|_1$ and by the triangle inequality 
  \begin{eqnarray*}    
    \|χ(A') - x^*\|_1 & \le & \|χ(A')  - χ(B)\|_1 + \|χ(B) - x^*\|_1 \\ 
                      &\le & (2mΔ)^{12 m}  ⋅ (mΔ) + m \\
                      & ≤ & (2mΔ)^{13 m} \qedhere
  \end{eqnarray*}
\end{proof}

\section{Sensitivity}
\label{sec:sensitivity}

We now show Theorem~\ref{thr:5}, the main result of this paper. We first show the $1$-dimensional case, $m=1$. The case $m≥2$ builds on the same important structural theorem on weighted matroids that is laid out in the first part of this section but also requires some further techniques from convex and polyhedral geometry.

\subsection{One constraint}
\label{sec:sens-with-single}

Recall that 
the  elements of the ground set  have  integer weights $w： Ε → 
\{-Δ,\dots,Δ\}$. 
We first explain how Theorem~\ref{thr:5}  follows from the following assertion on weighted matroids. 

\begin{theorem}
  \label{thr:2}
  Let $A,B ∈ \I$  be disjoint with $|A| = |B|$,  and let $w： E → \{-Δ , \dots , Δ\}$ be integer weights with  $|w(A) - w(B)|≤ μ$. If $|A| = |B| \ge (2⋅ Δ+1) ^5 + μ$ then there exists $A' \in \I$, $A' ≠ A$ 
  with
  \begin{displaymath}
    w(A') = w(A)  \quad \text{ and } \quad  |A'| = |A|.
  \end{displaymath}
\end{theorem}

\begin{proof}[Proof of Theorem~\ref{thr:5} for $m=1$] 
  Given a bases $A,B$, let $μ ∈ ℕ_0$ with $|w(A) - w(B)| ≤ μ$ and suppose that $A$ has smallest symmetric difference with  $B$ among all bases of weight $w(A)$. Consider $A' = A \setminus (A ∩B)$ and  $B' = B \setminus (A ∩B)$ and the minor $M' = ( A' ∪ B', \I')$, where
  \begin{displaymath}
    \I' = \{ I \setminus (A ∩B) ：I ∈ \I, \, (A ∩B) ⊆ I ⊆ (A ∪B) \}. 
  \end{displaymath}
  The rank of $M'$ is $|A\setminus B|$. If $|A \oplus B| ≥ (2 ⋅ Δ+1)^5 + μ$, then Theorem~\ref{thr:2} implies that there  exists a basis $C'$ of $M'$ different from $A'$ with weight $w(C') = w(A')$. This yields a basis
  \begin{displaymath}
    C' ∪ (A ∩ B)  \quad \text{of weight} \quad w(C' ∪ (A ∩ B)) = w(A)
  \end{displaymath}
  of the matroid $M$ that has more elements in common with $B$ than $A$. This is a contradiction to the minimality of the symmetric difference of $A$ and $B$. 
\end{proof}
\noindent 
The rest of this section is devoted to a proof of Theorem~\ref{thr:2}. We begin by showing a simple observation. 
\begin{proposition}
  \label{lem:2}
  Let $S$ be a finite set and $w ：S → \{-Δ, \dots, Δ\}$ be integer weights  with
  \begin{equation}
    \label{eq:1}
    |w(S)|  = μ. 
  \end{equation}
	Denote the set of non-negative elements by $S^+ = \{ s ∈ S ： w(s) ≥ 0\}$.  
  Let $S_1,\dots, S_ℓ$ be a partitioning of $S$ into subsets of $S$.  There exists an index $i$ such that
  \begin{equation*}
    \label{eq:2}
    |S^+ ∩S_i|  ≥ \frac{|S| -μ }{ ℓ ⋅ (Δ+1) }. 
  \end{equation*}
\end{proposition}
\begin{proof}
  Equation~\eqref{eq:1} implies
  \begin{displaymath}
    |S^+| ≥ (|S| - μ) / (Δ+1). 
  \end{displaymath}
By an averaging argument, there exists a set $S_i$ that contains at least $(|S| - μ) / (ℓ ⋅(Δ+1))$ many elements of $S^+$.   
\end{proof}

\begin{lemma}
  \label{thr:1}
  Let $A,B ∈ \I$  be disjoint with $|A| = |B| = k$ and let $w： E → \{-Δ , \dots , Δ\}$ be integer weights with  $|w(A) - w(B)|≤ μ$. There exist  subsets $A' ⊆ A$ and $B' ⊆ B$ of equal cardinality such that
  \begin{enumerate}[i)] 
  \item $A \setminus A' ∪ B' ∈ \I$, 
  \item $|A'| = |B'| ≥ (k - μ) / (2⋅ Δ+1)^2$  and 
  \item 
    $w(a)≥w(b)$ for each $a ∈ A'$ and $b ∈ B'$. 
  \end{enumerate} 
\end{lemma}

\begin{proof}
  Let $a_1,\dots,a_k$ be an ordering of $A$ such that $w(a_i) ≥ w(a_{i+1})$ for all $i$. Furthermore, let $A_\alpha = \{ a ∈ A： w(a) = \alpha\}$ for $\alpha=-Δ,\dots,Δ$. The $A_\alpha$ are  a partitioning of $A = A_{-Δ}  ∪ \cdots ∪ A_{Δ}$.  
  From this we construct a partitioning  $B_{-Δ}  ∪ \cdots ∪ B_{Δ}$ of $B$   such that, for each $j ∈ \{-Δ,\dots, Δ\}$ one has
  \begin{displaymath}
    B_j ∪ A_{j+1} ∪ \cdots ∪ A_Δ ∈ \I. 
  \end{displaymath}
  The existence of such a partitioning $B_{-Δ} ∪ \cdots ∪ B_{Δ}$ of
  $B$ is guaranteed by the exchange property of the matroid, specifically 
   to the sets
  \begin{equation}
    \label{eq:3}
    A_{j+1} ∪ \cdots ∪ A_Δ ∈ \I  \text{ and } B \setminus (B_{-\Delta} ∪ \cdots ∪ B_{j-1}) ∈ \I.
  \end{equation}  
  If $B_{-\Delta},\dots,B_{j-1}$ have been constructed, then $B_j ⊆ B$ is a
  subset of cardinality $|A_j|$ of the right independent set 
  in~\eqref{eq:3} that can be added to $A_{j+1} ∪ \cdots ∪ A_Δ$.  

   Now let $b_1,\dots,b_k$ be any ordering of $B$ such that the elements of
   $ B_{j}$ come  before the elements of $B_{j+1}$ for each
   $j$. 
   In the following, we refer to a tuple $a_ib_i$ as an \emph{edge}. 
  We apply Proposition~\ref{lem:2} to the set of edges  $S = \{ a_ib_i ： i=1,\dots,k\}$ and the weight function $w'： S → \{-2⋅Δ , \dots , 2⋅Δ\}$ defined by the difference of weights
  \begin{displaymath}
    w'(a_ib_i) = w(a_i) - w(b_i).
  \end{displaymath}
  The partitioning of $S$ is according to the value of the $a_i$. Formally,
  \begin{displaymath}
    S = S_{-Δ}  ∪ \cdots ∪ S_Δ, \quad \text{  where } \quad S_\alpha = \{a_ib_i ： w(a_i) = \alpha\}.
  \end{displaymath}
  Proposition~\ref{lem:2} now shows that there exists an index $j ∈ \{-Δ,\dots, Δ\}$ such that $S_j$ contains at least
  \begin{displaymath}
    \frac{k - μ  }{(2Δ+1)^2} 
  \end{displaymath}
	 non-negative edges $a_ib_i$.  Let $B' ⊆ B_j$ be the corresponding end-nodes
     of these edges on the side of $B$. The following is an independent set 
  \begin{displaymath}
    B' ∪ A_{j+1} ∪ \cdots ∪ A_{Δ},
  \end{displaymath}
  simply because $B_j ∪ A_{j+1} ∪ \cdots ∪ A_{Δ} ∈ \I$.
  Since $A$ is independent, there exists a subset $\wt{A} ⊆ A_{-Δ} ∪ \cdots ∪ A_{j }$ such that
  \begin{displaymath}
    \wt{A} ∪ B' ∪ A_{j+1} ∪ \cdots ∪ A_{Δ} 
  \end{displaymath}
  is an independent set of cardinality $|A| = k$. Let
  \begin{displaymath}
    A ' = A \setminus  (\wt{A}  ∪ A_{j+1} ∪ \cdots ∪ A_{Δ}). 
  \end{displaymath}
  Clearly, we have $|A'|  = |B'| ≥ ({k - μ  })/(2Δ+1)^2$  
  and
  \begin{displaymath}
    (A \setminus A') ∪ B' ∈ \I. 
  \end{displaymath}
  The crucial observation is that all elements in $A'$ have weight at least $j$ and all weights in $B'$ have weights at most $j$.
\end{proof}

A subset $S ⊆ E$ of the ground set is called \emph{unicolor}, if $w(x) = w(y)$ for each $x,y ∈ S$.  
The following is a version of \cref{thr:1} guaranteeing an exchange with unicolor sets. 
 \begin{corollary}
   \label{co:1}
   Let $A,B ∈ \I$  be disjoint with $|A| = |B| = k$ and let $w： E → \{-Δ , \dots , Δ\}$ be integer weights with  $|w(A) - w(B)|≤ μ$. There are unicolor subsets $A' ⊆ A$ and $B' ⊆ B$ of equal cardinality such that
   \begin{enumerate}[i)] 
   \item $A \setminus A' ∪ B' ∈ \I$,
   \item $|A'| = |B'| ≥ (k - μ) / (2⋅ Δ+1)^4$,
   \item $w(a)≥w(b)$ for each $a ∈ A'$ and $b ∈ B'$.
   \end{enumerate} 
 \end{corollary}

 \begin{proof}
	 Let $A^{(1)} ⊆ A$ and $B^{(1)} ⊆ B$ be the sets obtained from
	 \cref{thr:1}, which are not unicolor, but satisfy all other required properties.
	 Further, $|A^{(1)}| = |B^{(1)}| \ge (k - \mu) / (2\Delta+1)^2$.
	 At least $|B^{(1)}| / (2⋅Δ+1)$ elements of  $B^{(1)}$ have the same weight.
	 By Lemma~\ref{lem:3} we can thus restrict to this unicolor subset $B^{(2)}\subseteq B^{(1)}$ with a corresponding subset $A^{(2)}\subseteq A^{(1)}$ guaranteeing
  \begin{displaymath}
	  |A^{(2)}|  = |B^{(2)}| ≥ ({k - μ  })/(2Δ+1)^3.
   \end{displaymath}
	 Again, by restricting to the largest unicolor subset of $A^{(2)}$ we have $A^{(3)}\subseteq A^{(2)}$ and $B^{(3)}\subseteq B^{(2)}$ both unicolor and
    \begin{displaymath}
	    |A^{(3)}|  = |B^{(3)}| ≥ ({k - μ  })/(2Δ+1)^4. \qedhere
   \end{displaymath}
 \end{proof}

\begin{proof}[Proof of Theorem~\ref{thr:2}] 
  \cref{co:1} implies that there exist two unicolor sets
  $A^+ ⊆A$ and $B^+ ⊆ B$ of equal cardinality at least $2 ⋅ Δ$ such
  that
  \begin{itemize}
  \item $|A^+| = |B^+ |$, 
  \item $w(a) ≥ w(b)$ for each $a ∈A^+$ and $b ∈ B^+$, and
  \item $(A  ⧹ A^+ ) ∪ B^+ ∈ \I$. 
  \end{itemize}
By symmetry, there also exist unicolor sets $A^- ⊆A$ and $B^- ⊆ B$ of equal cardinality at least $ 2⋅ Δ$ such
  that
  \begin{itemize}
  \item $|A^-| = |B^- |$, 
  \item $w(a) ≤  w(b)$ for each $a ∈A^+$ and $b ∈ B^+$, and
  \item $(A  ⧹ A^- ) ∪ B^- ∈ \I$. 
  \end{itemize}
 Let $p  = w(a^+) - w(b^+) ∈ \{0,\dots,Δ\}$ be the weight of the edges  $a^+b^+$, $a^+ ∈ A^+$, $b^+ ∈ B^+$.
	Similarly, let $q ∈ \{0,\dots,Δ\}$ such that $-q$ is the common weight of the edges $a^-b^-$, $a^- ∈ A^-$, $b^- ∈ B^-$.
 If $p = 0$ then one has
  \begin{displaymath}
    w \left( (A  ⧹ A^+ ) ∪ B^+ \right) = w(A). 
  \end{displaymath}
  This is an independent set of cardinality $|A|$  and weight $w(A)$, which is what we want. 
  Similarly, when $q = 0$ then the assertion follows trivially. Hence, assume for the remainder of the proof that
	$p, q \ge 1$.

	Since all four sets are of cardinality at least $2 ⋅Δ$, we can assume, by downsizing (\cref{lem:3}) if necessary, that
  \begin{displaymath}
    |A^+| = |B^+| = 2q  \, \text{ and }   |A^-| = |B^-| = 2p . 
  \end{displaymath}
	We next consider the fractional point $y^*$
  \begin{eqnarray*}    
    y^*  & = &  \frac{1}{2} χ\left(  (A  ⧹ A^+ ) ∪ B^+\right) +  \frac{1}{2} χ\left(  (A  ⧹ A^- ) ∪ B^-\right) \label{eq:5} \\ 
	  & = &   χ(A) -  \frac{1}{2}  \left( χ(A^+) - χ(B^+) + χ(A^-) - χ(B^-) \right). \label{eq:4}        
  \end{eqnarray*}
Clearly, $y^* ∈ P_B(M)$ and the weight of $y^*$ is the same as the weight of $A$, i.e., 
  \begin{eqnarray*}    
    w\T y^* & = & w(A) - \frac{1}{2} (-p) ⋅2 ⋅ q - \frac{1}{2} q ⋅2 ⋅ p \\
           & = & w(A). 
  \end{eqnarray*}
  Since $A^+, B^+, A^-, B^- \neq \emptyset$ one has $y^* ≠ χ(A)$. 
  For $\alpha ∈ \{ -Δ, \dots , Δ\}$, let us denote the elements of weight $\alpha$ by
  $E_\alpha = \{ e ∈ E ： w(e) = \alpha\}$. We next argue that the sum of the
  components of $y^*$ corresponding to $E_\alpha$ are integral for each $\alpha
  ∈ \{-\Delta,\dots,Δ\}$. This follows from the fact that the sets $A^+,B^+,A^-$ and $B^-$ are unicolor and of even cardinality, implying that
  \begin{displaymath}
    χ(E_\alpha)\T \left( χ(A^+) - χ(B^+) + χ(A^-) - χ(B^-) \right) 
  \end{displaymath}
  is an even integer. Consequently one has for each $\alpha ∈ \{-Δ,\dots,Δ\}$
  \begin{displaymath}
    χ(E_\alpha)\T y^* ∈ ℕ_0. 
  \end{displaymath}
  We next consider the \emph{partition matroid} $M_p = (E, \I_p)$ with
  \begin{displaymath}
    \I_p = \left\{ I ⊆ E ： |I ∩ E_\alpha| ≤ χ(E_\alpha)\T y^*, \, \alpha ∈ \{-Δ,\dots,Δ\}      \right\}.
  \end{displaymath}
  The corresponding matroid polytope $P(M_p)$ is defined by the inequalities
  \begin{equation}
    \label{eq:6}
    \begin{array}{rcll}
      ∑_{e ∈ E_\alpha} x_e & ≤ &  ∑_{e ∈ E_\alpha}  y_e^*, & \alpha ∈ \{-Δ , \dots , Δ\} \\
                 x_e & ≥ & 0, &  e ∈ E.                  
    \end{array}
  \end{equation}
  The point $y^*$ satisfies all rank-constraints in~\eqref{eq:6} with equality.
  The crucial observation is now the following.
  \begin{quote}
    Each $y ∈ P_B(M) ∩ P_B(M_p)$ that satisfies the rank constraints in \eqref{eq:6} with equality is of weight $w(y) = w(A)$.  
  \end{quote}
  The \emph{matroid intersection polytope} $P_B(M) ∩ P_B(M_p)$ is integral~\cite{edmonds1979matroid,lawler1975matroid}, see also~\cite[Theorem~41.12]{schrijver2003combinatorial}. The fractional point $y^*$ is therefore in the convex hull of at least two integral vectors that are also tight at the rank inequalities in~\eqref{eq:6}. One of them corresponds to the characteristic vector $χ(A')$  of an independent set $A'$ different from $A$. The cardinality of this independent set is equal to the one of $A$, since the sum of the right-hand-sides of the rank constraints in~\eqref{eq:6}  is equal to $|A|$ and $χ(A')$ satisfies all these constraints with equality. 
\end{proof}

\subsection{Several constraints}
We start by stating two lemmas on cones generated by small discrete vectors.
Recall that the convex cone generated by a set of vectors $X \in \mathbb R^m$ is defined as
\begin{displaymath}
  \mathrm{cone}(X) = \left\{\sum_{x\in X} \lambda_x x : \lambda_x \ge 0,\ x\in X\right\} ⊆ ℝ^m.  
\end{displaymath}
A cone is \emph{pointed} if $\mathbf 0$ is a  vertex of the cone and \emph{flat} otherwise. 
\begin{lemma}\label{lem:hadamard}
  Let $X\subseteq \{-\Delta,\dotsc,\Delta\}^m$ such that
  $C = \mathrm{cone}(X)$ is pointed.  Then there is a halfspace
  $H = \{x\in\mathbb R^m : d\T x \ge 0\}$ with
  $C \cap H = \{\mathbf 0\}$ defined by some $d \in \mathbb Z^m$ with
	\begin{equation*}
		\lVert d \rVert_{\infty} \le \Delta^{m} m^{m/2+1} .
	\end{equation*}
\end{lemma}
\begin{proof}
  One may assume without loss of generality that $C$ has exactly
  $\dim(C) \le m$ facets.  Otherwise, the origin is a vertex of
  another cone containing $C$ that is defined by $\dim(C)$ facets of
  $C$ with linearly independent normals.  The vector $d$ can then be taken as
  the sum of the facet normals of $C$, each of which is up to scaling
  fully defined by being orthogonal to the $\dim(C)-1$ vectors in $X$
  spanning the facet.  By Cramer's rule and Hadamard's bound, for each
  face normal one can take an integer vector with entries bounded by
  $\Delta^m m^{m/2}$.
\end{proof}
\begin{lemma}\label{lem:caratheodory}
  Let $X\subseteq \{-\Delta,\dotsc,\Delta\}^m$ and $x\in X$ such that
  \begin{equation*}
    -x \in \mathrm{cone}(X\setminus\{x\})
    .
  \end{equation*}
  Then we can write
  $- \lambda_x x = \sum_{y\in X\setminus\{x\}} \lambda_{y} y$ for some
  $\lambda\in \mathbb Z_{\ge 0}^X$ with at most $m+1$ non-zero
  components and
	\begin{equation*}
	\| \lambda \|_\infty \le \Delta^m m^{m/2} .
	\end{equation*}
\end{lemma}
\begin{proof}
  By a variant of Carath\'eodory's Theorem, see~\cite[Theorem
  5.2]{schrijver2003combinatorial}, we may assume without loss of
  generality that $X\setminus \{x\}$ are linearly independent, in
  particular, $|X|\le m + 1$.  The assertion follows immediately from
  applying Cramer's rule and Hadamard's bound.
\end{proof}

\begin{figure}
	\centering
	\begin{tikzpicture}

    \pgfmathsetmacro{\phi}{(1+sqrt(5))/2}
		\draw[thick] (0, 1.25, 1.25) -- (0, 1.25, -1.25) -- (0, -1.25, -1.25) -- (0, -1.25, 1.25) -- cycle;
		\draw[->] (0, 0, 0) -- node[pos=1, right] {$\color{black}\cdots$} (1.5, -1, -1);
		\fill[lightgray] (0, 0, 0) -- (1.35, 0.9 * \phi, 0) -- (1.35, -0.9, -0.9) -- (1.35, -0.9, 0.9) -- cycle;
		\draw[->] (0, 0, 0) -- node[pos=1, right] {$\color{black}\delta_0$} (1.5, \phi, 0);
		\draw[gray] (1.35, -0.9, -0.9) -- (1.35, -0.9, 0.9) -- (1.35, 0.9 * \phi, 0) -- cycle;
		\draw[->] (0, 0, 0) -- node[pos=1, right] {$\color{black}\delta_{i}$} (1.5, -1, 1);
		\draw[->] (0, 0, 0) -- node[pos=1, left] {$\color{black}\delta_{i+1}$} (-0.5, -1, 0.5);
		\draw[->] (0, 0, 0) -- node[pos=1, left] {$d$} (-2, 0, 0);
		\node at (-0.5, 1.5, 0) {$H$};
	\end{tikzpicture}
	\caption{Visualization of the proof of Lemma~\ref{lem:several-exchanges}}
\end{figure}
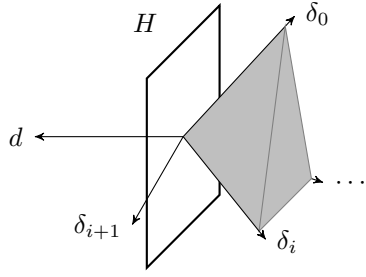

\noindent 
In analogy to the one-dimensional case, we call a set $S\subseteq E$ \emph{unicolor}, if $W(a) = W(b)$ for all $a,b\in S$. The following is a
multidimensional analogue of \cref{thr:1}.
\begin{lemma}\label{lem:m-monotone}
  Let $A, B\in \I$ be disjoint with $|A| = |B| = k$ and let
  $W: E\rightarrow \{-\Delta,\dotsc,\Delta\}^m$ be weight vectors with
  $\lVert W(A) - W(B) \rVert_1 \le \mu$.  Further, let
  $d\in\mathbb Z^m$. There exist unicolor sets $A'\subseteq A$,
  $B'\subseteq B$ of equal cardinality such that
\begin{enumerate}[i)] 
\item $A\setminus A' \cup B'\in \I$,
\item
  $|A'| = |B'| \ge \frac{k - \| d \|_1 \mu}{(2 \|d \|_1 Δ +
    1)^2  (2\Delta+1)^{2m}}$, and
\item $d\T W(a) \ge d\T W(b)$ for each $a\in A'$ and $b\in B'$.
\end{enumerate}
\end{lemma}
\begin{proof}
  We apply \cref{thr:1} for the single-dimensional weight function
\begin{displaymath}
  w(e) = d\T W(e), \, e\in E.
\end{displaymath}
Then $|w(A) - w(B)| ≤ \|d\|_1 ⋅ μ$ and we obtain sets
$A^{(1)}\subseteq A$ and $B^{(1)}\subseteq B$ of equal cardinality
with
\begin{enumerate}[1)]
\item $A\setminus A^{(1)}\cup B^{(1)} \in \I$,  \label{item:6}
\item $|A^{(1)}| = |B^{(1)}| \ge (k - \|d\|_1  μ ) / (2 \|d\|_1 Δ + 1)^2$, and \label{item:7}
\item $w(a) \ge w(b)$ for all $a\in A^{(1)}, b\in B^{(1)}$. \label{item:8}
 \end{enumerate}
 Note that there exists a $B^{(2)} \subseteq B^{(1)}$ with at least $|B^{(1)}|/(2\Delta+1)^m$ elements $e$ of
 the same weight $W(e)$. Via the downsizing (\Cref{lem:3}) we can
 restrict to $B^{(2)}$ and a corresponding subset of $A^{(2)}
 \subseteq A^{(1)}$ while still satisfying properties \ref{item:6}) and \ref{item:8}). Next, we observe that there exists $A^{(3)}
 \subseteq A^{(2)}$ with at least $|A^{(2)}|/(2\Delta+1)^m$ elements $e$ of
 the same weight $W(e)$. Downsizing again, we obtain two unicolor sets
 $A^{(2)}$, $B^{(2)}$ that satisfy properties \ref{item:6}) and \ref{item:8}) of cardinality 
 \begin{displaymath}
   |A^{(2)}| = |B^{(2)}| 
   \ge \frac{k - \|d\|_1 \mu}{(2 \|d\|_1 \Delta + 1)^2 \cdot (2\Delta+1)^{2m}} .
	\end{displaymath}
	Thus, the sets $A^{(2)}, B^{(2)}$ satisfy all required properties.
\end{proof}

\begin{lemma}\label{lem:several-exchanges}
  Let $A, B\in \I$ be disjoint with $|A| = |B| = k$, 
  and $\mu = \lVert W(A) - W(B) \rVert_1$.  Then there are sets
  \begin{equation}
    \label{eq:11}
          A_0,A_1,\dotsc,A_{\ell}\subseteq A \,\text{ and } B_0,B_1,\dotsc,B_{\ell}\subseteq B, 
        \end{equation}
        all unicolor,
such that
\begin{enumerate}[i)]
\item $A\setminus A_i\cup B_i \in \I$ for all $i$, \label{item:9}
\item $|A_i| = |B_i| \ge k / (2 m Δ)^{10 m} - \mu$ for all $i$, \label{item:10}
\item $-\delta_0 \in \mathrm{cone}(\{\delta_1,\dotsc,\delta_{\ell}\})$, where $\delta_i = W(a) - W(b)$ for all $a\in A_i, b\in B_i$. 
  \label{prop:conic}
\end{enumerate}
\end{lemma}
\begin{proof}
  We construct the sets iteratively.  $A_0 \subseteq A$ and
  $B_0 \subseteq B$ can be created from \Cref{lem:m-monotone} using
  $d = (0,0,\dotsc,0)$.

  Suppose we already have $A_0,\dotsc,A_i$ and $B_0,\dotsc,B_i$ and
  that these sets satisfy \ref{item:9}) and \ref{item:10}).  If
  $C = \mathrm{cone}(\{\delta_0,\dotsc,\delta_{i}\})$ is flat then it
  contains some non-zero $x, y$ with $x + y = 0$. In particular, there
  exists $\lambda\in\mathbb R_{\ge 0}^{i+1}$ with
  $\sum_{j=0}^i \lambda_j \delta_j = 0$ and $\lambda_k > 0$ for some
  $k\in\{0,1,\dotsc,i\}$.  After swapping $A_0, A_k$ and $B_0, B_k$ in \eqref{eq:11}, the sequence of sets also satisfies~\ref{prop:conic}).

  Assume now that $C$ is pointed.  From \Cref{lem:hadamard} it follows
  that for some $d\in\mathbb Z^m$ with
  $\lVert d\rVert_\infty \le (2\Delta)^{m} m^{m/2+1}$ the halfspace
  \begin{equation*}
    H = \{x \in \mathbb R : d\T x \ge 0\} 
  \end{equation*}
  intersects $C$ exactly in $0$.  By
  applying~\cref{lem:m-monotone}, we obtain unicolor sets
  $A_{i+1}, B_{i+1}$ satisfying \ref{item:9}) and
  \begin{align*}
    |A_{i+1}| = |B_{i+1}| &\ge \frac{k - \mu \|d\|_1}{(2\Delta \| d \|_1 + 1)^2(2\Delta+1)^{2m}} \\
                          &\ge \frac{k}{(4\Delta \| d \|_1)^2(4\Delta)^{2m}} - \mu \\
                          &\ge \frac{k}{(4\Delta (2Δ)^m m^{m/2 +2})^2(4\Delta)^{2m}} - \mu \\
                          &\ge \frac{k}{ (2 Δ m)^{10 m}} - \mu
  \end{align*}
	thus also satisfying~\ref{prop:conic}).  Furthermore, if $\delta_{i+1} = 0$, then, after swapping $A_0, A_{i+1}$ and $B_0, B_{i+1}$ in \eqref{eq:11}, the sequence of sets also satisfies~\ref{prop:conic}). Otherwise, 
  $\delta_{i+1}\in H$ must be different from
  $\delta_1,\dotsc,\delta_i \notin H$.  Since there are finitely many
  values for $\delta_i$, the procedure will ultimately terminate and
  therefore eventually satisfy~\eqref{prop:conic}.
\end{proof}

\subsubsection*{The proof of the multidimensional sensitivity theorem}

As in the one-dimensional case, \cref{thr:5} reduces to the
following statement by repeating the arguments
from~\cref{sec:sens-with-single}.
\begin{lemma}
  Let $A, B\in \I$ disjoint with $|A| = |B|$ and let
  $W: E\rightarrow \{-\Delta,\dotsc,\Delta\}^m$ be a multidimensional
  weight function with $\|W(A) - W(B)\|_1 \le \mu$ , where $μ ∈ ℕ_+$. If
  \begin{displaymath}
    |A| = |B| \ge (2mΔ)^{12 ⋅m} \mu 
  \end{displaymath}
  then there exists $A'\in \I$, $A'\neq A$ with
  \begin{displaymath}
    W(A') = W(A) \quad \text{ and } \quad |A'| = |A|.
  \end{displaymath}
\end{lemma}
\begin{proof}
  From Lemma~\ref{lem:several-exchanges} we obtain unicolor sets
  $A_i, B_i$, $i=0,1,\dotsc,\ell$ with
  \begin{enumerate}
  \item $A\setminus A_i \cup B_i \in \I$,
  \item $|A_i| = |B_i| \ge (2 m Δ)^{2 m} -1$,
  \item $-\delta_0 \in \cone(\{\delta_1,\dotsc,\delta_\ell\})$, where
    $\delta_i = W(b) - W(a)$ for each $b\in B_i$ and $a\in A_i$.
  \end{enumerate}
  Due to~\cref{lem:caratheodory} we may assume that $\ell\le m$ and
  there exists
  $\lambda\in\mathbb Z_{\ge 0}^{\ell+1} \setminus \{\mathbf 0\}$ with
  $\sum_{i=0}^\ell \lambda_i \delta_i = 0$ and
  $\|\lambda\|_\infty \le (2\Delta)^m m^{m/2}$.  Apply downsizing 
  (\cref{lem:3}) to each $A_i,B_i$, to obtain an arbitrary
  $A'_i \subseteq A_i$, $|A'_i| = (\ell+1) \lambda_i$ and a
  corresponding $B'_i\subseteq B_i$ with $|B'_i| = |A'_i|$ and
  $A\setminus A'_i \cup B_i\in\I$.  We proceed as in the case
  with a single equality constraint and refer the reader to it for details.
  It holds that
  \begin{equation*}
    y^* = \sum_{i=0}^\ell \frac{1}{\ell+1} \chi(A\setminus A'_i \cup B'_i)
  \end{equation*}
  is in $P_B(M)$, satisfies $\sum_{e\in E} y^*_e = \mathrm{rank}(M)$,
  $W y^* = W(A)$ and has an integral number of elements of each
  weight vector.  Thus, $y^*$ must be a convex combination of bases of
  $M$, all of which have weight $W(A)$ and since $y^* \neq \chi(A)$
  not all of them can be equal to $A$.
\end{proof}

\section{Applications}
\label{sec:applications}
In this section we give specific examples of problems that can be cast
as finding a basis of a matroid subject to $m$ constraints, each of
which have one of the following forms.
\begin{itemize}
	\item Equality constraint: given $w: E\rightarrow \{-\Delta,\dotsc,\Delta\}$ and $\beta\in\mathbb Z$, require $w(B) = \beta$.
	\item Inequality constraint: given $w: E\rightarrow \{-\Delta,\dotsc,\Delta\}$ and $\beta\in\mathbb Z$, require $w(B) \le \beta$, or alternatively $w(B) \ge \beta$.
	\item Congruence constraints: given $p\in\{1,2,\dotsc,\Delta\}$, $w: E\rightarrow \{0,1,\dotsc,p-1\}$, and $\beta\in\{0,1,\dotsc,p-1\}$, require $w(B) \equiv \beta \mod p$.
\end{itemize}
While we proved our FPT algorithm only for the first type, it is easy to reduce the other two to it.
This follows from standard constructions similar to slack variables.
To this end, consider a matroid $M = (E, \I)$.
Suppose for a given weight function $w: E\rightarrow \{-\Delta,\dotsc,\Delta\}$ and $\beta\in\mathbb Z$, we are searching
for a basis $B$ with $w(B) \le \beta$ and possibly other linear constraints (with one of the three types
from above).
We define $M'$ as the direct sum of $M$ and a rank $n$ uniform matroid with $2n\Delta$ elements. The weight function $w$
is extended by setting an arbitrary half of the elements in the uniform matroid to weight zero and the other half to weight~$1$. Any other linear constraint is extended with zero coefficients for the new elements, which means that they do not affect it. It follows easily that a basis $B$ of the original matroid $M$ can be extended to a basis $B'$ of $M'$ that satisfies $w(B') = \beta$ if and only if $w(B)\le \beta$.

Now suppose we are given $p\in\{1,2,\dotsc,\Delta\}$, $w: E\rightarrow \{0,1,\dotsc,p-1\}$, and $\beta\in\{0,1,\dotsc,p-1\}$ and are searching for a basis $B$ with $w(B) \equiv \beta \mod p$ and possibly other constraints.
Similar to before, we obtain $M'$ as the direct sum of $M$ with a uniform matroid of rank $n$ over
$2n$ elements. The weight function $w$ is extended such that $n$ many new
elements have weight $-p$ and $0$ each. Again, any other linear constraint is
extended with zero coefficients for the new elements. Then a basis $B$ of $M$ is extendible to a basis $B'$ of $M'$ that satisfies $w(B') = \beta$ if and only if $w(B) \equiv \beta \mod p$.

\medskip

Very recently, matroid problems with labels from an abelian group have gained some attention. Liu and
Xu~\cite{liu2023congruency} study the following problem.

\defproblem{Group-Constrained Matroid Base}
{Matroid $M = (E,\I)$, a 
labelling $\psi : E \rightarrow \Gamma$ for an abelian group $(\Gamma, \odot)$, and $g \in \Gamma$.}
{Find base $B$ of $M$ with $g = \psi(B) := \bigodot_{b\in B} \psi(b)$}

Liu and Xu prove that if $\Gamma = \mathbb{Z}_m$ and $m$ is
either product of two primes or a prime power, then the problem can be solved in
FPT time in $m$. Our result generalizes this to all finite abelian groups. 

\begin{corollary}
  If $(\Gamma, \odot)$ is a finite abelian group then Group-Constrained
  Matroid Base can be solved in $f(m)\cdot n^{O(1)}$ time, for
  $m = |\Gamma|$.
\end{corollary}
\begin{proof}
	Every finite abelian group is isomorphic to the direct product of cyclic groups. Thus, we may
	assume without loss of generality that $\Gamma = \mathbb Z_{m_1} \times \mathbb Z_{m_2} \times \cdots \times \mathbb Z_{m_\ell}$, where $m = m_1 \cdot \dotsc \cdot m_{\ell}$ for prime powers $m_1,m_2,\dotsc,m_{\ell}\le m$. 
	We can therefore model the problem using $\ell \le \log_2 m$ congruency constraints with entries bounded by $m$.
\end{proof}
Similar to our results, Liu and Xu's techniques are based on
proximity statements. However, their techniques rely on specific groups.
More precisely, Liu
and Xu~\cite{liu2023congruency} prove that their techniques would work if
a certain conjecture in additive combinatorics is true. The conjecture 
is proven when $m$ is either product of two primes or a prime power, which allows
them to obtain the result. In that matter, our techniques allow us to circumvent
this issue.

Motivated by Liu and Xu's work, H\"orsch et al.~\cite{horsch2024problems} considered the
problem in the setting of non-finite groups. Among
other problems, they show an FPT algorithm for the Group-Constrained
Matroid Problem with $g=0$ parameterized by $|\Gamma|$ in a special
case when the matroid is $\text{GF}(q)$-representable for a prime power
$q$. Similarly, to~\cite{liu2023congruency} their techniques also rely on
an additive combinatorics result of Schrijver and
Seymour~\cite{schrijver1990spanning}, which prohibits their techniques from
working in general finite abelian groups.

\medskip

Budgeted matroid problems, in which one has to find an independent set subject to one or more budget constraints and possibly maximizing a profit function,
have been studied with a great extend towards approximation schemes, see e.g.~\cite{grandoni2010approximation, doron2023eptas, doron2023budgeted}. Finding bases is generally at least as hard as finding independent sets, since one
can always fix the cardinality of the solution therefore reducing to bases.
In the area of FPT algorithms, Marx~\cite{marx2009parameterized} devised an algorithm for these type of problems,
specifically motivated by the problem of finding a feedback edge set.
Given a graph $G(V,E)$ a feedback edge set is a subset $X$ of edges such that $G(V,E\setminus X)$ is acyclic.

\defproblem{Feedback Edge Set with Budget Vectors}
{A graph $G = (V,E)$, a vector $W(e) \in \mathbb Z_{\ge 0}^m$ for each $e \in E$, a
budget $b \in \mathbb{Z}_{\ge 0}^m$.}
{Find a minimum cardinality feedback edge set $X$ such that $W(X) \le b$.}

Note that for $m=1$, this is a weighted variant of Feedback Edge Set which can
be solved in polynomial time by a Greedy algorithm~\cite{marx2009parameterized},
however, for unbounded $m$ and $\Delta = \lVert W \rVert_\infty$ the problem is
NP-hard. Marx~\cite{marx2009parameterized} presented a randomized FPT algorithm
in the parameters $m$, $\Delta$, and $|X|$. A direct application of our theorem
is that a weaker parameterization that depends on just $m$ and $\Delta$
suffices.

\begin{corollary}
    Feedback Edge Set with Budget Vectors can be solved in $(m\Delta)^{O(m^2)}
    \cdot n^{O(1)}$ randomized time.
\end{corollary}
\begin{proof}
	The problem is equivalent to finding a spanning forest $F$ with $W(F) \ge W(E) - b$, which can be solved
	in the mentioned time using~\cref{thm:alg}.
\end{proof}

Fairness considerations have inspired new variants of many problems, where
elements belong to different groups and each group needs to be represented in
the solution to some level~\cite{bolukbasi2016man,celis2018ranking,bera2014approximation, bandyapadhyay2019constant,
anegg2022technique, jia2022fair}. For example
Abdulkadiro{\u{g}}lu and S\"onmez~\cite{abdulkadirouglu2003school} address the assignment of
students to schools, where fairness constraints are selected to achieve
racial, ethnic, and gender balance. For similar models see also~\cite{kamada2012stability, sankar2021matchings}.

\defproblem{Matching with Group Fairness Constraints}
{A bipartite graph $G = (A\cup B,E)$, where each $b\in B$ belongs to a set of groups $G(b)\subseteq\{1,2,\dotsc,m\}$, and
a quote $q_i \in \mathbb{Z}_{\ge 0}$ for each group $i=1,2,\dotsc,m$.}
{Find a matching containing at least $q_i$ many elements $b\in B$ with $i\in G(b)$ for $i=1,2,\dotsc,m$.}

To the best of our knowledge, FPT algorithms have not been considered for this problem before.
\begin{corollary}
    Matching with Group Fairness Constraints can be solved in $(m\Delta)^{O(m^2)} \cdot n^{O(1)}$ randomized time.
\end{corollary}
\begin{proof}
	Consider a transversal matroid defined on $G$ with elements $B$. For each $b\in B$ let $W(b)$ be the indicator vector of $G(b)$ and $\Target = q$.
\end{proof}

Partition matroids subject to equality constraints have been
studied in the context of block structure integer linear programming,
albeit without an explicit reference to matroids.
Specifically, this variant was studied under the name \emph{combinatorial $n$-fold}~\cite{knop2020combinatorial}. 
This has led to exponential improvements in the running time of FPT
algorithms for problems in computational social choice, in string problems and
more, see~\cite{knop2020combinatorial} for an overview.  One concrete problem,
that is captured by a partition matroid with equality constraints is 
the Closest String Problem. In this problem, one is given $m$ strings, and the
goal is to compute a string that minimizes the maximum Hamming distance to any
of the input strings. Gramm et al.~\cite{gramm2003fixed} design an
FPT algorithm parameterized by $m$ (see also~\cite{knop2020combinatorial}).
Using our result, we can obtain an FPT algorithm for
the following more general problem.

\defproblem{Closest Base}
{A matroid $M = (E,I)$ with $n$ elements and a subset of bases $S := \{B_1,\ldots,B_m\}$.}
{Find the basis $B$ (not necessarily in $S$) that minimizes $\max_i |B \oplus B_i|$.}

\begin{corollary}
    Closest Base on representable matroids can be solved in $m^{O(m^2)} n^{O(1)}$ randomized time.
\end{corollary}
This running time matches the best known running time for the Closest String
problem~\cite{knop2020combinatorial}.
\begin{proof}
	It is equivalent to minimize $H = \max_i |B\setminus B_i| = 1/2 \cdot \max_i |B \oplus B_i|$.
    We start by guessing $H$. Note that $H \le |E| / 2$, so only polynomially many guesses are
	required. The weight vector $W(e) \in \{0, 1\}^m$ is defined as $W(e)_i = 0$ if $e \in B_i$ and $1$
	otherwise. Further, we define the target vector $\Target = (H, H,\dotsc,H)$. We then use~\cref{thm:alg} to find
	a basis $B$ with $W(B) \le \Target$.
\end{proof}

\section{Lower bound for matroid intersection}
\label{sec:intersection}
\begin{figure}
	\centering
	\begin{tikzpicture}
		\node(v1)[draw, circle, inner sep = 0.5mm] at (0, 0) {};
		\node(u1)[draw, circle, inner sep = 0.5mm] at (0, 1) {};
		\node(v2)[draw, circle, inner sep = 0.5mm] at (1.5, 0) {};
		\node(u2)[draw, circle, inner sep = 0.5mm] at (1.5, 1) {};
		\node(d) at (2.25, 0.5) {\dots};
		\node(v3)[draw, circle, inner sep = 0.5mm] at (3, 0) {};
		\node(u3)[draw, circle, inner sep = 0.5mm] at (3, 1) {};
		\draw (v1) -- (u1) node[pos = 0.5, left] {$1$};
		\draw (v2) -- (u2) node[pos = 0.5, left] {$0$};
		\draw (v1) -- (u2) node[pos = 0.25, below] {$0$};
		\draw (u1) -- (v2) node[pos = 0.25, above] {$0$};
		\draw (v2) -- (d);
		\draw (u2) -- (d);
		\draw (d) -- (v3);
		\draw (d) -- (u3);
		\draw (v3) -- (u3) node[pos = 0.5, right] {$0$};
	\end{tikzpicture}
\quad
	\begin{tikzpicture}
		\node(v0)[draw, circle, inner sep = 0.5mm] at (-1.5, 0) {};
		\node(u0)[draw, circle, inner sep = 0.5mm] at (-1.5, 1) {};
		\node(v1)[draw, circle, inner sep = 0.5mm] at (0, 0) {};
		\node(u1)[draw, circle, inner sep = 0.5mm] at (0, 1) {};
		\node(v2)[draw, circle, inner sep = 0.5mm] at (1.5, 0) {};
		\node(u2)[draw, circle, inner sep = 0.5mm] at (1.5, 1) {};
		\node(d) at (2.25, 0.5) {\dots};
		\node(v3)[draw, circle, inner sep = 0.5mm] at (3, 0) {};
		\node(u3)[draw, circle, inner sep = 0.5mm] at (3, 1) {};
		\draw (v0) -- (u0) node[pos = 0.5, left] {$1$};
		\draw (v0) -- (u1) node[pos = 0.25, below] {$0$};
		\draw (u0) -- (v1) node[pos = 0.25, above] {$0$};
		\draw (v1) -- (u1) node[pos = 0.5, left] {$1$};
		\draw (v2) -- (u2) node[pos = 0.5, left] {$0$};
		\draw (v1) -- (u2) node[pos = 0.25, below] {$0$};
		\draw (u1) -- (v2) node[pos = 0.25, above] {$0$};
		\draw (v2) -- (d);
		\draw (u2) -- (d);
		\draw (d) -- (v3);
		\draw (d) -- (u3);
		\draw (v3) -- (u3) node[pos = 0.5, right] {$0$};
	\end{tikzpicture}
	\caption{Example of high proximity and sensitivity in exact matroid intersection}
	\label{fig:intersection}
\end{figure}
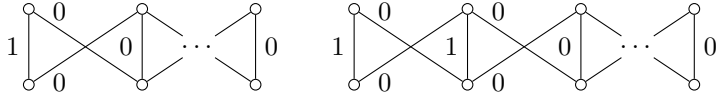
In this section, we remark that low proximity and sensitivity do not generalize
to matroid intersection, even for $m = 1$ and $\Delta = 1$. Our examples are
based on matchings in a bipartite graph $G = (A\cup B,E)$. Although matching in
a bipartite graph does not form an independent set in a single matroid, it can
be represented as a matroid intersection of two partition matroids, where one
partition matroid restricts the degree of each vertex in $A$ to be at most one
and the other does the same for $B$.

We start with the example of a high sensitivity for matroid intersection.
\begin{theorem}\label{thm:sensitivity-example}
	For infinitely many $n \in \mathbb N$ there exist matroids $M = (I, E), M' = (I', E)$
    over the same set of elements and a weight function $w : E \rightarrow \{0, 1\}$,
    such that there exist exactly two 
    two common bases $B$ and $B'$ of both matroids that satisfy:
	\begin{enumerate}
	 \item $w(B) = 0$ and $w(B') = 1$,
	 \item $B \cap B' = \emptyset$ and $|B| = |B'| = n/2$.
	\end{enumerate}
\end{theorem}
\begin{proof}
For $n$ even, we create an instance of bipartite matching consisting of a cycle of length $n$ that has a single edge of weight $1$ and all others of weight $0$, see also left cycle in Figure~\ref{fig:intersection}.
There are exactly two perfect matchings, which correspond to the common bases of the two matroids
and trivially satisfy the properties stated in the theorem.
\end{proof}
Next, we show an example of a high proximity for matroid intersection.
\begin{theorem}
	For infinitely many $n \in \mathbb N$ there exist matroids $M = (I, E), M' = (I', E)$
    and a weight function $w : E \rightarrow \{0, 1\}$ with
	a vertex solution $x^*$ to the continuous relaxation $x\in P_B(M)\cap P_B(M'), w\T x = 1$ such that
	there is a unique common basis $B$ of both $M$ and $M'$ with $w(B) = 1$ and $B$ satisfies
	$\lVert x^* - \chi(B) \rVert_1 = 3/4 \cdot n$.
\end{theorem}
\begin{proof}
	For $n$ a multiple of $4$,
consider an instance of bipartite matching similar to the proof
of~\cref{thm:sensitivity-example}, but now 
on a graph with two disjoint cycles of length $n/2$ each, each of which we can think of as the union of two perfect matchings.
The first cycle has a single edge of weight $1$ and the second cycle has two edges of weight $1$, which
appear in the same perfect matching.
All other edges have weight zero, see also Figure~\ref{fig:intersection}.
Now suppose we want to find a perfect matching with total weight equal to $1$.
The only such perfect matching, i.e., the only common basis of both matroids, takes the perfect matching
with the weight $1$ edge in the first cycle and the all-zero perfect matching in the second cycle.
This is also a vertex solution for the continuous relaxation, but there is a second vertex solution
that takes each edge of the second cycle $1/2$ times and the all-zero perfect matching
of the first cycle. 
The distance between the two solutions is $3/4 \cdot n$.
\end{proof}

\bibliographystyle{plain}
\bibliography{references}

\end{document}